\documentclass[12pt]{iopart}

\usepackage{amssymb, setstack}

\usepackage{graphicx}
\usepackage{graphics}
\newtheorem{Theorem}{Theorem}[section]
\newtheorem{Proposition}{Proposition}[section]

\newtheorem{Lemma}{Lemma}[section]

\def\proof{\par{\it Proof}. \ignorespaces}
\def\endproof{{\ \vbox{\hrule\hbox{%
     \vrule height1.3ex\hskip0.8ex\vrule}\hrule }}\par}
\newenvironment{Proof}{\proof}{\endproof}

\newtheorem{Remark}[Theorem]{Remark}

\newlength{\myVSpace}
\newlength{\bigmyVSpace}
\newcommand\bigxstrut{\raisebox{\bigmyVSpace}
  {\rule{0pt}{\bigmyVSpace}}%
}

\eqnobysec

\begin{document}


\title{The Pfaff lattice on symplectic matrices}

\author{Yuji Kodama$^1$ and Virgil U. Pierce$^2$}
\address{$^1$ Department of Mathematics, Ohio State University, Columbus,
OH 43210}
\ead{kodama@math.ohio-state.edu}
\address{$^2$ Department of Mathematics, University of Texas -- Pan American, Edinburg, TX  78539}
\ead{piercevu@utpa.edu}


\begin{abstract}
The Pfaff lattice is an integrable system arising from the SR-group
factorization in an analogous way to how the Toda lattice arises
from the QR-group factorization.  In our recent paper [{\it Intern.
Math. Res. Notices}, (2007) rnm120], we studied the Pfaff lattice hierarchy
for the case where the Lax matrix is defined to be a lower Hessenberg matrix.
 In this paper we deal with the case of a symplectic lower Hessenberg Lax matrix, this forces the Lax matrix to take a tridiagonal shape.
We then show that
the odd members of the Pfaff lattice hierarchy are trivial, while the even members
are equivalent to the indefinite Toda lattice hierarchy defined in
[Y. Kodama and J. Ye, {\it Physica D}, {\bf 91} (1996) 321-339].
This is analogous to the case of the Toda lattice hierarchy in the relation to the
Kac-van Moerbeke system.
 In the case with initial matrix having only real or imaginary eigenvalues, the fixed points of the even flows are given by $2\times 2$ block diagonal matrices with zero diagonals.
 We also consider a family of skew-orthogonal polynomials with symplectic recursion relation related to the Pfaff lattice, and find that they are succinctly expressed in terms of orthogonal polynomials appearing in the indefinite Toda lattice.
\end{abstract}


\maketitle


\pagestyle{empty}
\tableofcontents



\section{Introduction}

The Toda lattice hierarchy of nonlinear differential equations may be formulated as the continuous limit of the QR-algorithm for diagonalizing a symmetric matrix \cite{rutishauser:58, deift:83, symes:82}.  There exists a number of generalizations of the QR-algorithm for diagonalizing matrices while preserving different symmetries, the HR and SR-algorithms are two that we will consider in this paper \cite{benner:97, bunse:86B, Fa93, rutishauser:58}.  In particular the SR-algorithm relies on SR-factorizations whose goal is to factor an invertible matrix into a symplectic matrix and an upper triangular matrix \cite{bunse:86}.  The Pfaff lattice hierarchy was recently introduced to describe the partition functions of the orthogonal and symplectic ensembles of random matrices, and equivalently, to give the evolution of skew-orthogonal polynomials (see \cite{adler:99, adler:02}).  The Pfaff lattice hierarchy may be viewed as a continuous limit of the SR-algorithm for diagonalizing a lower Hessenberg matrix, and in fact was originally discovered in this context \cite{watkins:88}.  We showed in our recent work \cite{kodama:07} that the Pfaff lattice is an integrable system in the Arnold-Liouville sense, and in the case that the initial conditions are not lower Hessenberg matrices from the symplectic algebra, $\mathfrak{sp}(n)$,  we showed that the fixed points of the Pfaff flows with an initial condition having real distinct eigenvalues are given by $2\times 2$ block diagonal matrices with zero diagonals.

In this paper we will consider the Pfaff lattice flow with an initial condition given by a matrix from the symplectic algebra $\mathfrak{sp}(n)$.    The Lie algebra $\mathfrak{sp}(n)$ is the set of $2n \times 2n$ matrices satisfying $J L + L^T J = 0$ with
\[ J = \mbox{diag}_2\left(\,
\left( \begin{array}{cc} 0&1 \\ -1 &0 \end{array} \right),\,\dots,\, \left(\begin{array}{cc} 0 & 1 \\-1 & 0 \end{array}\right) \,\right) = I_n \otimes \left( \begin{array}{cc} 0 & 1 \\ -1 & 0 \end{array}\right) \,. \]
Here diag$_2(A)$ denotes the $2\times 2$ block diagonal matrix of $A$, i.e. ${\rm diag}_2(A) = {\rm diag}_2(A_1, \ldots, A_2)$ with the $2\times 2$ blocks $A_j$ in the diagonal of $A$, and $I_n$ is the $n\times n $ identity matrix.
\begin{Remark}
Matrices in the symplectic algebra, $\mathfrak{sp}(n)$, are called Hamiltonian in the linear algebra and
 control theory literature, we will avoid this nomenclature to avoid confusion with the Hamiltonians
 which are constants of motion and appear throughout the paper.  The term symplectic algebra is also
 the one commonly used in the literature on the Pfaff lattice and in representation theory \cite{fulton:91}.
\end{Remark}
The spectrum of a real symplectic matrix comes in three types:
\begin{itemize}

\item[(a)] Pairs of real eigenvalues $(z, -z)$.

\item[(b)] Paris of imaginary eigenvalues $(z, -z)$.

\item[(c)] Or quadruples of complex eigenvalues $(z, \bar{z}, -
  z, -\bar{z})$.

\end{itemize}
Symplectic matrices appear in a number of control theory problems most
 notably in
 solving the real algebraic Riccati equation (see e.g. \cite{Fa00}).
  Our $J$ is a permutation of the traditional $J$, say $\tilde{J}$, used in the
  literature on the Hamiltonian eigenvalue problem, i.e.
  $\tilde{J}=\left(\begin{array}{cc} 0&1\\-1&0\end{array}\right) \otimes I_n$.  We use this $J$
  because of its connection to the Pfaff lattice which is defined on
  lower Hessenberg matrices (see below).

The Pfaff lattice
 can be viewed as an $\mathfrak{sp}$-version of the
Toda lattice, it has the following form with $\mathfrak{sp}$-projection \cite{adler:99, kodama:07},
\[
\frac{\partial L}{\partial t}=[\pi_{\mathfrak{sp}}(\nabla H),\, L]\,.
\]
The Pfaff lattice is associated with the
Lie-Poisson structure induced by the splitting $\mathfrak{sl}(2n) =
\mathfrak{k} \oplus \mathfrak{sp}(n)$. Recall that the splitting
gives the following projections of an element
$X\in{\mathfrak{sl}}(2n,{\mathbb R})$,\\
\begin{eqnarray} \label{projection1}
\pi_{\mathfrak k}(X)&=X_--J\left(X_+\right)^TJ+\frac{1}{2}(X_0-JX_0^TJ)\\
\label{projection2}
\pi_{\mathfrak{sp}}(X)&=X_++J\left(X_+\right)^TJ+\frac{1}{2}(X_0+JX_0^TJ)
\end{eqnarray} \\
where $X_0$ is the $2\times 2$ block diagonal part of $X$, $X_+$ (resp. $X_-$) is
the $2\times 2$ block upper (resp. lower) triangular part of $X$. In particular, we have
\begin{equation}
\label{mathfrak_k}
 \mathfrak{k} = \left\{ \left( \begin{array}{cccc}
 h_1 I_2 & 0_2 &  \dots & 0_2 \\
\star & h_2 I_2  & \dots & 0_2 \\
   \vdots &   \vdots  & \ddots  & \vdots   \\
\star      &   \star      &  \cdots   & h_n I_2 \end{array}\right)\,: \,
\sum_{k=1}^n h_k = 0 \right\} \,, \end{equation} where $I_2$ is the
$2\times 2$ identity matrix, and $0_2$ is the $2\times 2$ zero
matrix.

As in the case of the Toda lattice hierarchy, one can also define
 the Pfaff lattice hierarchy
 (see \cite{kodama:07}, and also Section \ref{pfaff}),
\begin{equation}\label{Pfaff}
 \frac{\partial L}{\partial t_j} = \left[ \pi_{\mathfrak{sp}}(\nabla H_j), L\right]\, \quad{\rm for}\quad j=1,\ldots,2n-1\,,
\end{equation}
where $H_j=\frac{1}{j+1}{\rm tr}(L^{j+1})$ and $L$ is a Hessenberg matrix given by,
for example with $n=3$,
\begin{equation}\label{L}
L=\left( \begin{array}{cccccc}
0 &  c_1  &   0 & 0 & 0 & 0 \\
** & b_1 & a_1 & 0 & 0 & 0 \\
** & * & -b_1& c_2 & 0 & 0 \\
** & * & * & b_2 & a_2 & 0\\
** & * & * & * &-b_2 & c_3 \\
** & * & * & * & * & 0
\end{array}\right) \,.
\end{equation}
The $c_k$ are  invariants of the Pfaff lattice, in the traditional Pfaff lattice consider by Adler et. al. \cite{adler:99} the $c_k$ are chosen to be $1$.  It was shown in \cite{kodama:07} that the $c_k$ are Casimirs of the Pfaff lattice.
In \cite{kodama:07} it was shown that for $L$ with distinct
eigenvalues
 the Pfaff lattice hierarchy is
an integrable system in the Arnold-Liouville sense. In addition it
was shown that if $L$ is not symplectic and has distinct real
eigenvalues, then as $t_j \to \infty$ for any $j$, there is a
diagonal matrix $P(t_j)$ such that the normalized matrix $
\hat{L}(t_j):=P(t_j) L(t_j) P(t_j)^{-1}$
approaches a $2\times 2$ block upper triangular matrix, where the
diagonal blocks are sorted by the size of $z^j$ for the real
eigenvalues $z$ appearing in each.
In this paper, we deal with the Pfaff lattice when the initial matrix is
symplectic, and discuss the dynamical structure in a connection
with the Toda lattice.

A natural question at this point is then what happens when we
restrict the Pfaff lattice variables to matrices which are both
symplectic and lower Hessenberg.  The first consequence is that such
matrices have the compact form:
\begin{equation}\label{BG.L0form}
L_0=
\left( \begin{array}{ccccccc}
\begin{array}{cc} 0 & c_1 \\ d_1 & 0 \end{array} &\vline&
\begin{array}{cc} 0  & 0 \\ a_1 & 0 \end{array} &\vline& \cdots &\vline& 0_2
\\\hline\bigxstrut
\begin{array}{cc} 0 & 0 \\ a_1 & 0 \end{array} &\vline&
\begin{array}{cc} 0 & c_2 \\ d_2 & 0 \end{array} &\vline&\cdots &\vline& 0_2
\\\hline\bigxstrut
\vdots &\vline& \vdots &\vline& \ddots &\vline& \vdots
\\\hline\bigxstrut
0_2 & \vline& 0_2 &\vline& \cdots &\vline &\begin{array}{cc} 0 & c_n \\ d_n & 0\end{array}
\end{array}\right) \,,
\end{equation}
that is, $L_0$ has a $2\times 2$ block tridiagonal form
$ L_0 = (l_{i,j})_{1\le i,j\le n} $ with
 $2\times 2$ block matrices $l_{i,j}$ having $l_{i,j}=0_2$ for $|i-j|>1$ and
\[
l_{k,k}=\left(\begin{array}{cc} 0&c_k \\d_k & 0\end{array}\right),\quad
l_{k,k+1}=l_{k+1,k}=\left(\begin{array}{cc} 0&0\\a_k&0\end{array}\right)\,,
\]
where $c_k = \pm 1$. We call a matrix in this form
$H^S$-tridiagonal.  This is analogous to restricting the Toda
lattice variables to a matrix which is tridiagonal and
skew-symmetric (i.e. in $\mathfrak{so}(n)$); a situation we will
remark on in more detail in Section \ref{pfaffToda}.

We will show that the even Pfaffian flow of an $H^S$-tridiagonal matrix
converges to a sequence of block diagonal matrices with the
following shapes:
\begin{itemize}

\item[(a)] a $2\times 2$ block for each real pair of eigenvalues $(z,-z)$,

\item[(b)] a $2\times 2$ block for each imaginary pair of eigenvalues $(z,-z)$,

\item[(c)] a $4 \times 4$ block for each quadruple of complex eigenvalues $(z,\bar{z},-z,-\bar{z})$.

\end{itemize}
In addition, for the $t_{2j}$ flow, the blocks are sorted by the
size of $\mbox{Re}( z^{2j})$.  In the appendix, we will give a
simple proof that the SR-algorithm is given by the integer
evaluation of the flow generated by the Pfaff lattice introduced in
\cite{adler:99} (see also \cite{kodama:07}).


The present paper is organized as follows:

 In Section \ref{SRalgorithm} we review the SR-factorization which is closely associated with the
 Pfaff lattice hierarchy of equations.
 We then show, in Section \ref{pfaff}, that the odd members of the hierarchy of the Pfaff lattice
 (\ref{Pfaff}) on
 symplectic matrices are trivial, while the even members of the hierarchy on  $H^S$-tridiagonal
 matrices are equivalent to the indefinite Toda hierarchy introduced in
\cite{KoYe96} (Theorem \ref{TP}).
 This is similar to the case of the Toda lattice in the connection with
 the Kac-van Moerbeke system.
In contrast with the symmetric Toda lattice hierarchy, the indefinite Toda hierarchy experiences
 blow ups, where some
 entries of the matrix approach infinity.  Our result then implies that for
 generic $H^S$-tridiagonal matrices the Pfaff lattice has blow ups.

 In Section \ref{pfaffToda}, we give a further discussion on the equivalence
between the Pfaff lattice on $H^S$-tridiagonal matrix and the indefinite Toda lattice defined in
\cite{KoYe96}.  We begin by reviewing the indefinite Toda lattice,
then show that the $\tau$-functions for the Pfaff lattice are the same
as those for the indefinite Toda lattice (Theorem \ref{PTtau}).
 We then examine the families of
skew-orthogonal polynomials appearing in the Pfaff lattice for
an $H^S$-tridiagonal matrix.  We show that these polynomials
are related to the
orthogonal polynomials in the indefinite Toda lattice
(Theorem \ref{SO-O relation}).
Finally, we remark on the
asymptotic behavior of the even Pfaff lattice flows of an
$H^S$-tridiagonal matrix.  In particular we show that initial conditions with
complex eigenvalues, or with complex eigenvectors, will result in a blow up.

\section{Background on SR-factorization}\label{SRalgorithm}

The Pfaff lattice equation is intimately connected to the SR-factorization of an invertible matrix and so we collect here some relevant facts on this factorization.
With the Lie subalgebra $\mathfrak{k}$ given by (\ref{mathfrak_k}) in the Lie algebra
splitting $\mathfrak{sl}(2n)=\mathfrak{k}\oplus\mathfrak{sp}(n)$,
we define $G_{\mathfrak{k}}$ to be the Lie group with Lie
algebra $\mathfrak{k}$,
\[ G_{\mathfrak{k}}:=\left\{~
\left( \begin{array}{cccc}
\alpha_1 I_2 & 0_2 &  \cdots & 0_2 \\
\star & \alpha_2 I_2 &  \cdots & 0_2 \\
  \vdots    &  \vdots       & \ddots &   \vdots  \\
  \star    &   \star      & \cdots  & \alpha_n I_2 \end{array}\right)\, :\, \prod_{j=1}^n\alpha_j=1~\right\}\,.
\]
We also define the group $\tilde{G} \supset G_{\mathfrak{k}}$ to be the
group of real invertible lower $2\times 2$ block triangular matrices
with {\em free} invertible diagonal blocks, that is, $\tilde{G}$ is a parabolic subgroup of lower
$2\times 2$ block matrices of $SL(2n,\mathbb{R})$.

We define the Pfaffian of a skew-symmetric matrix $m$ by the recursive formula
\[ \mbox{pf}(m) = \sum_{j=1}^{2n} (-1)^{i+j+1} m_{ij}
\mbox{pf}(m_{\hat{i}\hat{j}}) \,,\] where $m_{\hat{i}\hat{j}}$ is
found by deleting the $i$-th and $j$-th rows and columns of $m$,
where $i$ is any row, and the initial value for the recursion is
given by
\[ \mbox{pf} \left(\begin{array}{cc} 0 & m_{12} \\ -m_{12} & 0
\end{array}\right)  = m_{12} \,.\]
Then, for example,
with the skew symmetric matrix
\[ m = \left(\begin{array}{cccc} 0 & m_{12} & m_{13} & m_{14} \\
                         & 0      & m_{23} & m_{24} \\
                         &        & 0      & m_{34} \\
                         &        &        & 0 \end{array}\right)
\]
we have
\[ \mbox{pf}(m) = m_{12} m_{34} - m_{13} m_{24} + m_{14} m_{23}\,. \]
Throughout this paper, we leave blank the lower triangular part of skew-symmetric matrices.

We will use the
SR-factorization of
$g \in SL(2n, \mathbb{R})$  introduced by Bunse-Gerstner
(Theorem 3.8 in \cite{bunse:86}, see also \cite{Bu82}):
%
%
\begin{Theorem}\label{tautheorem}
Let $g\in SL(2n, \mathbb{R})$ and $M = g J g^T$. Then
the factorization
\[ g = r s\, \]
with $r \in \tilde{G} $ and $s\in Sp(n, \mathbb{R})$
exists if and only if the Pfaffian of the $2k\times 2k$ upper left submatrix of $M$,
denoted by $M_{2k}$, does not vanish, i.e.
\[
  {\rm pf}(M_{2k}) \neq 0\,,\qquad 1\leq k \leq n-1\,.
\]
\end{Theorem}
As a consequence there is a dense set of matrices in $SL(2n,
\mathbb{R})$ for which this decomposition exists.
If we restrict the factorization to using elements of $G_\mathfrak{k}
\subset \tilde{G}$ the Theorem becomes:
\begin{Theorem} \label{SRtheorem}
Let $g \in SL(2n, \mathbb{R})$ and $M=g J g^T$.  Then the
factorization
\[ g = r s\,\]
with $r\in G_{\mathfrak{k}}$ and $s\in Sp(n, \mathbb{R})$ exists if
and only if all the Pfaffians ${\rm pf}(M_{2k})$ satisfy the positivity condition,
\[
{\rm pf}(M_{2k}) > 0\,,  \qquad 1\leq k \leq n-1\,.
\]
\end{Theorem}
This implies that there is an open set of matrices in $SL(2n,
\mathbb{R})$ for which this decomposition exists.

Theorem 3.7 in \cite{bunse:86} shows that the set of matrices in
$SL(2n, \mathbb{C})$ with SR-factorization is neither open nor dense in $SL(2n,
\mathbb{C})$.  For this reason we will restrict our consideration to
the real groups for now.

In Remark 3.9 in \cite{bunse:86} it is noted that for the choice of
subgroups of $SL(2n, \mathbb{R})$ in that paper the decomposition
$g=rs$ is not unique,
as one may pass a factor from
\[ \tilde{G} \cap Sp(n, \mathbb{R}) = \left\{ \mbox{diag}_2\left( A_1,
A_2, \dots, A_n\right)\, : \det(A_j) = 1\,~~\forall j\,
 \right\}\,,
\]
between $r$ and $s$. In the case of $SR$ factorization with $R$
restricted to $G_{\mathfrak{k}}$ we find that factorizations are
unique up to a factor from
\[ G_{\mathfrak{k}} \cap Sp(n, \mathbb{R}) = \left\{ \mbox{diag}_2\left(
\pm I_2, \pm I_2, \dots, \pm I_2\right) \right\}\,.\] One notes that
$G_{\mathfrak{k}}$ is in fact made up of $2^n$ connected components,
and if we further restrict $R$ to the component of
$G_{\mathfrak{k}}$ containing the identity matrix we obtain unique
factorizations. Equivalently we obtain unique factorizations by
asking that the diagonal elements of $r$ be positive.  In contrast
the group $\tilde{G}$ is connected.

The Pfaffians defined in Theorem \ref{tautheorem} are also called the $\tau$-functions,
which play the fundamental role in the Pfaff lattice hierarchy \cite{adler:99, kodama:07}, i.e.
for $M=gJg^T$ and $M_{2k}$, the $2k\times 2k$ upper left submatrix of $M$,
\begin{equation}\label{tau}
\tau_{2k}:={\rm pf}\,\left(M_{2k}\right),\qquad k=0,1,\ldots, n\,,
\end{equation}
with $\tau_0=1$. Derivatives of the $\tau$-functions for
$g(\mathbf{t}) := \exp\left( \sum_{j=1}^{2n-1} t_j L_0^j\right) $, with $L(0)
= L_0$,
 generate the matrix entries of $L(\mathbf{t})$ in the form (\ref{L}), and
the solution of the Pfaff lattice hierarchy. For example, we have
\cite{adler:99, kodama:07} (see Section \ref{pfaff}),
\begin{equation}\label{ab}
a_k(\mathbf{t})=a_k(0)\frac{\sqrt{\tau_{2k+2}(\mathbf{t})
    \tau_{2k-2}(\mathbf{t}) }}{\tau_{2k}(\mathbf{t})},\quad\quad b_k(\mathbf{t})=\frac{\partial}{\partial t_1}\ln\tau_{2k}(\mathbf{t})\,.
\end{equation}
We also have the following Lemma:
\begin{Lemma} \label{sr}
Let $g = r s$ be a $2n\times 2n$ matrix where $r \in
{G}_\mathfrak{k}$ and $s \in Sp(n, \mathbb{R})$. Then
the $2\times 2$ block diagonal entries of $r$,
${\rm diag}_2(r)={\rm diag}_2(r_1I_2,\ldots,r_nI_2)$, are expressed by the $\tau$-functions,
\[ r_{k}
        = \sqrt{ \frac{\tau_{2k}}{\tau_{2k-2}} }
\,.\]
\end{Lemma}
\begin{proof} Since $M=gJg^T=rsJs^Tr^T=rJr^T$, we have
$\tau_{2k}={\rm pf}(M_{2k})=\prod_{j=1}^kr_j^2$, which leads to the formula.
\end{proof}
A necessary condition for the existence of the factorization is that $\tau_{2k} = \mbox{pf}(M_{2k})$ must be positive (Theorem \ref{SRtheorem});
Lemma \ref{sr} and the convention that $\tau_0 = 1$
 explain this condition.
There are explicit algorithms for carrying out SR-factorization and they have many features in common with those developed for QR-factorization, we refer the reader to \cite{Fa00} and \cite{bunse:86}.

\section{Pfaff lattice hierarchy}\label{pfaff}

In \cite{adler:99}, Adler, Horozov, and van Moerbeke introduced the
Pfaff lattice hierarchy to describe the partition functions of the
Gaussian orthogonal and symplectic ensembles of random matrices (GOE
and GSE), and to describe the evolution of the recursion relations
of skew-orthogonal polynomials. The finite Pfaff lattice hierarchy is a set
of Hamiltonian flows on $\mathfrak{sl}(2n, \mathbb{R})$ with the
Lie-Poisson structure induced by the splitting $ \mathfrak{sl}(2n,
\mathbb{R}) = \mathfrak{k} \oplus \mathfrak{sp}(n)$, and it is
defined as follows (see also \cite{kodama:07}):
 For any $X, Y\in\mathfrak{sl}(2n)$ we first
define the $\mathcal R$-bracket
\[ [X, Y]_{\mathcal{R}} = [ \mathcal{R}X, Y] + [X, \mathcal{R}Y]\,, \]
where $\mathcal{R}$ is the $R$-matrix given by $\mathcal{R} = \frac{1}{2} (\pi_{\mathfrak{k}} -
\pi_{\mathfrak{sp}}) $.
Then the Lie-Poisson bracket for any functions $F$ and $H$ on ${\mathfrak{sl}}^*(2n)\cong {\mathfrak{sl}}(2n)$ is defined by
\[ \left\{ F, H\right\}_\mathcal{R}(L) = \langle L, \left[ \nabla F, \nabla H
  \right]_\mathcal{R} \rangle\,, \]
  where $\langle A,B\rangle={\rm tr}(AB)$ and $\nabla F$ is defined by
  $\langle X,\nabla F\rangle=\frac{d}{d\epsilon}F(L+\epsilon X)|_{\epsilon=0}$.
The Pfaff lattice with respect to a Hamiltonian function $H(L)$ is defined by
\[
 \frac{dL}{dt} =\left\{H(L),L\right\}_\mathcal{R}(L)\,.
 \]
In particular, if the Hamiltonian is $Sp$-invariant, one can write
\[
\left\{H(L),L\right\}_\mathcal{R} (L)=  \left[ \pi_\mathfrak{sp}(\nabla H(L) ), L
  \right]\,.
\]
The traditional Hamiltonians giving the Pfaff lattice are $\frac{1}{k+1} \mbox{tr}(L^{k+1})$ \cite{kodama:07}.
The Pfaff lattice hierarchy is then defined by
\begin{equation}\label{pfaffhierarchy}
\frac{\partial L}{\partial t_k}=[\pi_{\mathfrak{sp}}(L^k), L]\qquad  k=1,\ldots,2n-1.
\end{equation}
Thus the $t_k$-flow of the Pfaff lattice hierarchy is associated to the Hamiltonian $H_k=\frac{1}{k+1}{\rm tr}(L^{k+1})$.
Each flow is solved by the following factorization procedure:
Factor
\begin{equation}\label{RSfactor}
\exp\left( t \nabla H(L(0)) \right) = R(t) S(t)
\end{equation}
with the initial matrix $L(0)$ using SR-factorization with $R$ in the connected component of
$G_{\mathfrak{k}}$ containing the identity
and $S\in Sp(n)$, then the solution is given by
\begin{equation}\label{LRS}
L(t) = R(t)^{-1} L(0) R(t)  = S(t) L(0) S(t)^{-1}\,.
\end{equation}
With this in mind, one sees that as in the case of the QR-algorithm  with the Toda lattice, the
SR-algorithm is given by integer evaluations of the Pfaff flow with
Hamiltonian $H(L) = {\rm{tr}}( L \ln(L) - L ) $ a result we will show in detail in the Appendix.

In \cite{kodama:07}, we showed that the Pfaff lattice hierarchy (\ref{pfaffhierarchy})
 is an integrable
system in the Arnold-Liouville sense.
Normalizing the matrix $L$  of the form (\ref{L}) by
\[ \hat{L} = P L P^{-1} \]
with the diagonal matrix $P$ in the $2\times 2$ block form,
\[ P = \mbox{diag}_2\left(I_2,\,  a_1I_2, \, (a_1 a_2)I_2,\, \dots \,,\,
\left(\prod_{j=1}^{n-1} a_j\right)I_2 \right)\,,
\]
$\hat{L}$ is  a lower Hessenberg matrix with $1$'s on the super diagonal.
Then we showed that if $L$ was not symplectic and all the eigenvalues are
 {\it real} and {\it distinct}, then
 as $t_j\to \infty$, $\hat{L}(t_j)$ converges to a $2\times 2$
block upper triangular matrix such that the diagonal blocks are sorted by the
size of $ z^j$.

Solutions of the Pfaff lattice hierarchy are generated by the
$\tau$-functions (introduced as obstructions to the SR-factorization in
the previous section).  They are found by the following procedure:
We first consider the factorization of $g(\mathbf{t}) := \exp\left( \sum_{j=1}^{2n-1} t_j L^j_0 \right)$
with $L_0=L(0)$ (see (\ref{RSfactor})),
\[
g(\mathbf{t})=R(\mathbf{t})S(\mathbf{t})\qquad{\rm with}\quad R(\mathbf{t})\in G_{\mathfrak k},~~S(\mathbf{t})\in Sp(n)\,.
\]
Then the skew-symmetric matrix $M(\mathbf{t})=g(\mathbf{t})Jg^T(\mathbf{t})$ becomes
\[
M(t)=R(\mathbf{t})JR^T(\mathbf{t})\,.
\]
Since $R\in G_{\mathfrak k}$, we have
\[
{\rm diag}_2(R):={\rm diag}\left(r_1I_2, r_2I_2,\ldots, r_nI_2\right)\,.
\]
Then the $\tau$-functions $\tau_{2k}$ defined in (\ref{tau}) can be written by
\[
\tau_{2k}(\mathbf{t})={\rm pf}(M(\mathbf{t})_{2k})=\prod_{j=1}^k r_j(\mathbf{t})^2\,.
\]
Then from (\ref{LRS}), i.e. $R(\mathbf{t})L(\mathbf{t})=L(0)R(\mathbf{t})$, we have
\[
a_k(\mathbf{t})=a_k(0)\frac{r_{k+1}(\mathbf{t})}{r_k(\mathbf{t})}=a_k(0)\frac{\sqrt{\tau_{2k+2}(\mathbf{t})\tau_{2k-2}(\mathbf{t})}}{\tau_{2k}(\mathbf{t})}\,,
\]
which gives the $a_k$'s in (\ref{ab}) (see also \cite{adler:99, kodama:07}).

\subsection{Odd Pfaff flows of symplectic matrices}\label{Oddflow}
Here we show that  if
$L(0)$ is a symplectic matrix then it is a fixed point of the
odd members of the Pfaff lattice hierarchy.
To see this, one notes:
\begin{Lemma}\label{lemma-3.1}
For $L\in\mathfrak{sp}(n)$, the odd power $L^{2j-1}$ is
also symplectic.
\end{Lemma}

\begin{proof}
Being symplectic is equivalent to $ J L J = L^T $.
Suppose that $J L^{2j-3} J = \left( L^T \right)^{2j-3}$.
 Then we have
\[ J L^{2j-1} J = J L^{2j-3} J J L JJ L J = \left( L^T \right)^{2j-3}
L^T L^T = \left( L^T \right)^{2j-1}
\]
so that the lemma is true by induction.
\end{proof}
Therefore, $\pi_{\mathfrak{sp}}(L^{2j-1}) = L^{2j-1}$ for $L\in\mathfrak{sp}(n)$,
hence the odd members of the Pfaff hierarchy become trivial, i.e.
\[ \frac{\partial L}{\partial t_{2j-1}} = [ L^{2j-1}, L ] = 0 \,.\]
Note in particular that all $b_k$ in (\ref{ab}) vanish, which is consistent with
the form $L$ in Theorem \ref{BG.hessenbergL0}, that is, the diagonal elements are all zero.

We note that this also happens for the Toda lattice hierarchy:
Recall that the Toda lattice equation for a symmetric matrix is based on the Lie algebra splitting
$\mathfrak{sl}(n) = \mathfrak{b} \oplus \mathfrak{so}(n)$,
where $\mathfrak{b}$ is the set of upper triangular matrices.
With the pairing $\langle A, B\rangle:={\rm tr}(AB)$ for $A,B\in\mathfrak{sl}(n)$,
we have $\mathfrak{sl}(n)\cong\mathfrak{sl}(n)^*=\mathfrak{b}^*\oplus\mathfrak{so}(n)^*$
where $\mathfrak{b}^*\cong\mathfrak{so}(n)^{\perp}={\rm Sym}(n)$ and
$\mathfrak{so}(n)^*\cong\mathfrak{b}^{\perp}=\mathfrak{n}$. Here
${\rm Sym}(n)$ is the set of symmetric matrices, and $\mathfrak{n}$ is
the set of strictly upper triangular matrices.
Then the Lie-Poisson bracket for the functions $F,G$ on ${\rm Sym}(n)$ is defined by
\[ \left\{ F, H \right\}(L) = \langle L, \left[ \nabla F, \nabla
  H\right] \rangle\quad{\rm for}\quad L\in{\rm Sym}(n)\,.\]
 The Toda lattice hierarchy for $L\in {\rm Sym}(n)$ is then defined by
\begin{equation}\label{symmetricToda}
\frac{\partial L}{\partial t_j}=\{H_j,L\}(L)=[\pi_{\mathfrak{so}}(\nabla H_j),L]\,.
\end{equation}
With the Hamiltonian functions $H_j(L)=\frac{1}{j+1}{\rm tr}(L^{j+1})$
for $j=1,\ldots,n-1$,  the differential equation for $L$ becomes
\begin{equation}
\frac{\partial L}{\partial t_j} = [ \pi_{\mathfrak{so}}(L^j), L],
\end{equation}
with $\pi_{\mathfrak{so}}(X)$
meaning the projection of the matrix $X$ on the $\mathfrak{so}(n)$ component.

Now if we extend the Toda lattice equation for the general $L\in
\mathfrak{sl}(n,\mathbb{R})$, an analogue of Lemma \ref{lemma-3.1}
is true with $\mathfrak{so}(n)$ instead of $\mathfrak{sp}(n)$. Then,
for $L\in\mathfrak{so}(n)$ (i.e. skew-symmetric), the odd members of
the hierarchy become trivial. However the even powers of $L$ are
symmetric, and the even members of the generalized Toda hierarchy
for $L$ are equivalent to symmetric tridiagonal Toda lattices.
 This case includes
the Kac-van Moerbeke system \cite{KM:75}: Consider for example the following $2n\times 2n$ skew-symmetric tridiagonal matrix,
\[
L=\left( \begin{array}{ccccc}
0  &  \alpha_1  &  0  &   \cdots  &  0   \\
-\alpha_1 &  0  &  \alpha_2 & \cdots  &  0  \\
\vdots  &   \ddots  &  \ddots  &  \ddots  &  \vdots  \\
0         &    \cdots   &   -\alpha_{2n-2} &   0  &  \alpha_{2n-1}  \\
0      &   \cdots  &  0   &   -\alpha_{2n-1}  &  0
\end{array}\right)\,.
\]
Then the Kac-van Moerbeke hierarchy may be expressed by $\displaystyle{
\frac{\partial L}{\partial t_{2j}}= [\pi_{\mathfrak{so}}(L^{2j}),\,L]}$, where the first member gives
\[
\frac{\partial \alpha_k}{\partial t_2}=\alpha_k(\alpha_{k-1}^2-\alpha_{k+1}^2),\qquad k=1,\ldots,2n-1\,,
\]
with $\alpha_0=\alpha_{2n}=0$.
We then note that the square $L^2$ is a symmetric matrix,
\begin{equation} \label{3A}
L^2=T^{(1)}\otimes \left(\begin{array}{cc} 1 & 0\\0&0\end{array}\right) \,+\, T^{(2)}\otimes \left( \begin{array}{cc} 0&0\\0&1\end{array}\right)\,,
\end{equation}
where $T^{(i)}$, for $i=1,2$, are $n\times n$ symmetric tridiagonal matrices given by
\[
T^{(i)}= \left(\begin{array}{ccccc}
b^{(i)}_1    & a^{(i)}_1 & 0 & \cdots & 0  \\
a^{(i)}_1 &  b_2^{(i)} & a_2^{(i)} & \cdots  & 0\\
\vdots    &   \vdots   &  \ddots   &  \ddots   &   \vdots  \\
0     &           0        &             \cdots  &     b_{n-1}^{(i)}    &   a_{n-1}^{(i)}\\
0     &           0       &          \cdots  &     a_{n-1}^{(i)} &b_{n}^{(i)}
\end{array}\right)\,,
\]
with $a^{(1)}_k = \alpha_{2k-1}
\alpha_{2k} $, $b^{(1)}_k = - \alpha_{2k-2}^2 - \alpha_{2k-1}^2 $,
$a^{(2)}_k=\alpha_{2k}\alpha_{2k+1}$, and
$b^{(2)}_k=-\alpha_{2k-1}^2-\alpha_{2k}^2$.  Then
one can show that each $T^{(i)}$ gives the symmetric Toda lattice, that is,
the equation $\displaystyle{\frac{\partial L}{\partial t_{2j}}=[\pi_{\mathfrak{so}}(L^{2j}),L]}$
splits into two Toda lattices,
\[
\frac{\partial T^{(i)}}{\partial t_{2j}} = [\pi_{\mathfrak{so}}(T^{(i)})^{j},
  T^{(i)}]\,\quad{\rm for}\quad  i=1,2\,.
\]
The equations for $T^{(i)}$ are connected by the Miura-type transformation, with the functions
 $(a^{(i)}_k,b^{(i)}_k)$, through the Kac-van Moerbeke variables $\alpha_k$ (see \cite{GHSZ:93}).

We now show that the even Pfaff flows on a $H^S$-tridaigonal matrix have
a similar structure.

\subsection{Even Pfaff flows of symplectic matrices}
We first show that the $H^S$-tridiagonal form (\ref{BG.L0form}) is invariant
under the even members of the Pfaff hierarchy. Then we show
 that the Pfaff flow for the matrix $L$ in the form (\ref{BG.L0form})
is related to the indefinite Toda lattice defined in \cite{KoYe96}.
The indefinite Toda lattice is a continuous version of the HR
algorithm \cite{Fa93}, and is defined as follows (see also
Section \ref{pfaffToda} for some details): Let $C$ be the
$n\times n$ diagonal matrix given by,
\[
C={\rm diag}(c_1,c_2,\ldots,c_n)\,,
\]
where $c_k$'s are the elements appearing in the diagonal blocks of (\ref{BG.L0form}),
and they are invariant under the Pfaff flow (see \cite{kodama:07}).
We also define $\tilde T$ by $\tilde T:=CT$ with the triadiagonal matrix (\ref{Tsymm}), i.e.
\begin{equation}\label{T}
\tilde{T} = \left(\begin{array}{ccccc}
c_1 d_1 & c_1 a_1 & \cdots  &\cdots & 0 \\
c_2 a_1 & c_2 d_2 & \cdots & \cdots & 0 \\
  \vdots  &  \vdots   & \ddots & \ddots & \vdots \\
  0  &  0 & \cdots& c_{n-1}d_{n-1} & c_{n-1} a_{n-1} \\
   0 & 0 & \cdots & c_n a_{n-1}& c_n d_n \end{array}\right)\,.
\end{equation}
Then the indefinite Toda lattice hierarchy is defined by
\begin{equation}\label{indefinite}
\frac{\partial\tilde{T}}{\partial t_{2j}}=[B_j,\tilde{T}]\,,
\end{equation}
where $B_j:=[\tilde{T}^j]_+-[\tilde{T}^j]_-$. Here the $[\tilde{T}^j]_{\pm}$
represent the projections on the upper $(+)$ and lower $(-)$ triangular parts
of the matrix $\tilde{T}^j$ (sometimes we write $\tilde{T}^j_{\pm}=[\tilde{T}^j]_{\pm}$).

 With $L$ in the $H^S$-tridiagonal form (\ref{BG.L0form}), we first note
\begin{equation}\label{L2Tij}
L^2=\tilde{T}\otimes\left(\begin{array}{cc} 1 & 0\\0&0\end{array}\right) \, +\,
\tilde{T}^T\otimes\left(\begin{array}{cc} 0&0\\0&1\end{array}\right)\,,
\end{equation}
  from which we have, for any even power,
\[  L^{2j} = \tilde{T}^j\otimes \left(\begin{array}{cc} 1&0\\0&0\end{array}\right)\,+\,(\tilde{T}^{j})^T\otimes\left(\begin{array}{cc}
0&0\\0&1
\end{array} \right)\,.\]
\begin{Remark}
We remark here that equation (\ref{L2Tij}) is similar to (\ref{3A}),
this structure is fundamental and implies that our problem will have
a similar structure to that of the Kac-van Moerbeke lattice.
\end{Remark}
Then using the projection (\ref{projection2}), we obtain
 \begin{equation}\label{2mgenerator}
  \pi_{\mathfrak{sp}}(L^{2j}) = B_j\otimes\left(\begin{array}{cc}1&0\\0&0\end{array}\right)\,
  - \, B_j^T\otimes\left(\begin{array}{cc}0&0\\0&1\end{array}\right)\,,
 \end{equation}
 with $B_j=\tilde{T}^j_+-\tilde{T}^j_-$.
   With (\ref{2mgenerator}), one can easily show that the $H^S$-tridiagonal form
   is invariant under the even members of the Pfaff lattice hierarchy.
 It is also immediate to see from (\ref{2mgenerator}) that the $2m$-th member of the Pfaff
lattice hierarchy (\ref{Pfaff}) in terms of $L^2$ is equivalent to the indefinite Toda lattice (\ref{indefinite}).
Thus we have the following theorem:
\begin{Theorem}\label{TP}
The even Pfaff flow for $L\in\mathfrak{sp}(n,\mathbb{R})$ in the $H^S$-tridiagonal
matrix form (\ref{BG.L0form}),
\[
\frac{\partial L}{\partial t_{2m}}=\left[\pi_{\mathfrak{sp}}\left(L^{2m}\right),L\right]\,,
\]
is equivalent to the indefinite Toda flow for $\tilde{T} = CT$ with the symmetric tridiagonal matrix $T$
of (\ref{Tsymm}) and $C:={\rm diag}(c_1,c_2,\ldots,c_n)$,
\begin{equation}\label{Ttoda}
 \frac{\partial \tilde{T}}{\partial t_{2m}} = [ B_m,\tilde{T}]\,,
\end{equation}
where $B_m=\tilde{T}_+^m-\tilde{T}_-^m$.
\end{Theorem}

If the $c_k = 1$ for all $k$, then  $\tilde{T} = T$ is symmetric and
$B_m=\pi_{\mathfrak{so}}(T^m)$. Theorem \ref{TP} then says that the
$2m$-th Pfaff flow is just the traditional symmetric tridiagonal
$m$-th flow of the Toda lattice hierarchy.  In terms of the
SR-algorithm in Proposition \ref{pfaffalgorithm}, we obtain as a
Corollary of Theorem \ref{TP} that the even-integer iterates,
$L_{2k}$ of the SR-algorithm are equivalent to the $\tilde{T}_k$
iterates of the HR-algorithm. This result is shown directly in
\cite{benner:97}, the key is the structure of equation
(\ref{L2Tij}).  With the permuted $J$ matrix used in
\cite{benner:97}, i.e. $J=\left(\begin{array}{cc}0&1\\-1&0\end{array}\right)\otimes I_n$,  the right hand side of relation (\ref{L2Tij}) takes
the form
\[ \left(\begin{array}{cc} \tilde{T} & 0_n \\ 0_n & \tilde{T}^T
\end{array}\right)=\left(\begin{array}{cc}1&0\\0&0\end{array}\right)\otimes \tilde{T}+\left(\begin{array}{cc}0&0\\0&1\end{array}\right)\otimes \tilde{T}^T\,.
\]

\section{Pfaff lattice vs indefinite Toda lattice} \label{pfaffToda}
Here we give a further discussion on the equivalence between the Pfaff lattice hierarchy on
$H^S$-tridiagonal matrix and the indefinite Toda lattice hierarchy introduced in \cite{KoYe96}.
\subsection{The indefinite Toda lattice}
Let us begin with a brief description of the indefinite Toda lattice defined by (\ref{indefinite})
for a matrix $\tilde{T}=CT$ with a symmetric matrix $T$ of the form (\ref{Tsymm}), i.e.
\[
\frac{\partial\tilde{T}}{\partial t}=[B,\tilde{T}]\qquad{\rm with}\qquad B=\tilde{T}_+-\tilde{T}_-\,.
\]
The solution $\tilde{T}(t)$ with the initial matrix $\tilde{T}_0$ can be solved by
the HR-factorization,
\[
\tilde{g}=\exp(t\tilde{T}_0)=r(t)\,h(t)\,,
\]
where $r$ is a lower triangular matrix and $h$ satisfies $hCh^T=C$ (recall that if $C=I_n$,
$h\in SO(n)$). It is then easy to show that the solution $\tilde{T}(t)$ is obtained by
\begin{equation}\label{Tequation}
\tilde{T}(t)=r(t)^{-1}\,\tilde{T}_0\,r(t)=h(t)\,\tilde{T}_0\,h(t)^{-1}\,.
\end{equation}
We note that the entries $a_k$ in $\tilde{T}$ are expressed in terms of the diagonal
elements of $r$, say ${\rm diag}(r)={\rm diag}(r_1,\ldots,r_n)$,
\[
a_k(t)=a_k(0)\,\frac{r_{k+1}(t)}{r_k(t)}\,.
\]
To find those diagonal entries from $\tilde{g}$, one considers the following matrix called
the moment matrix,
\begin{equation}\label{Tmoment}
M^{\rm{Toda}}:=\tilde{g}\,C\,\tilde{g}^T=r\,C\,r^T\,.
\end{equation}
Then the matrix $r$ can be found by the Cholesky factorization method.
Now the $\tau$-functions are defined by
\[
\tau_k^{\rm{Toda}}={\rm det}\,(M^{\rm{Toda}}_k)\,,
\]
where $M^{\rm{Toda}}_k$ is the $k\times k$ upper left submatrix of $M^{\rm{Toda}}$.
With (\ref{Tmoment}), we have
$
\tau_k^{\rm{Toda}}=\prod_{j=1}^kc_j\,r_j^2\,,
$
which gives $c_kr_k^2=\tau_k/\tau_{k-1}$. Then the entries $a_k$ of $\tilde{T}$ can be
expressed in terms of the $\tau$-functions,
\begin{equation}\label{Todatau}
a_k(t)=a_k(0)\,\sqrt{\frac{c_k}{c_{k+1}}}\,\frac{\left[{\tau^{\rm{Toda}}_{k+1}(t)\,\tau^{\rm{Toda}}_{k-1}(t)}\right]^{1/2}}{\tau^{\rm{Toda}}_k(t)}\,.
\end{equation}
As will be shown in the next section (see Theorem \ref{PTtau} below),
those are the same as the formulae for the $a_k$'s in (\ref{ab})
given in terms of the $\tau$-functions of the Pfaff lattice.

Let us also discuss the orthogonal functions appearing in the indefinite Toda lattice.
First we note that the Lax form (\ref{indefinite}) is given by the compatibility of the equations,
\[
\tilde{T}\Phi=\Phi D \qquad{\rm and}\qquad \frac{\partial\Phi}{\partial t_{2j}}=B_j\Phi\,,
\]
where $D={\rm diag}(\lambda_1,\ldots,\lambda_n)$ with the
eigenvalues $\lambda_k$, and $\Phi=(\phi_i(\lambda_j))_{1\le i,j\le
n}$ the eigenmatrix. As the orthogonality condition of the
eigenvectors with a normalization, we have
\begin{Lemma}\label{Torthogonal}
\[
\Phi C\Phi^T=C\,.
\]
\end{Lemma}
\begin{Proof}
We consider the following relation which is equivalent to the orthogonality relation,
\[
\Phi^T C^{-1} \Phi=C^{-1}\,.
\](This is sometimes called the second orthogonality relation.)
To show this, we note
\[\begin{array}{llllll}
 \lambda_j\phi(\lambda_i)^TC^{-1}\phi(\lambda_j)&=&\phi(\lambda_i)^TC^{-1} \tilde{T}\phi(\lambda_j)
 =\phi(\lambda_i)^TT\phi(\lambda_j)\\
  &=& (\tilde{T}\phi(\lambda_i))^T C^{-1}\phi(\lambda_j)=\lambda_i\phi(\lambda_i)^T C^{-1}\phi(\lambda_j)\,.
  \end{array}
  \]
Since we assume $\lambda_i\ne\lambda_j$, the matrix $\Phi^TC^{-1}\Phi$ is diagonal.
We also note that this matrix is invariant under the flow, i.e. $\partial (\Phi^TC^{-1}\Phi)/\partial t=0$.
Normalizing the diagonal matrix gives the result.
\end{Proof}
We then define the inner product for functions $f(\lambda)$ and
$g(\lambda)$,
\begin{equation}\label{Tproduct}
\langle fg\rangle_{\rm Toda}:=\sum_{k=1}^n
f(\lambda_k)g(\lambda_k)c_k\,,
\end{equation}
which defines a discrete measure $d\mu(\lambda)=\sum_{k=1}^n
c_k\delta(\lambda-\lambda_k)\,dz$. Lemma \ref{Torthogonal} implies
\[
\langle \phi_i\phi_j\rangle_{\rm{Toda}}=c_i\delta_{ij}\,.
\]
With the orthogonality relation, the moment matrix $M^{\rm Toda}$ can be expressed by
\begin{equation}\label{Mtoda}
M^{\rm Toda}=\Phi_0e^{2tD}C\Phi^T_0=\left(\langle \phi^0_i\phi_j^0e^{2tz}\rangle_{\rm Toda}\right)_{1\le i,j \le n}\,,
\end{equation}
where $\Phi_0$ is the initial eigenmatrix with the eigenvector
$\phi_i^0=\phi_i^0(\lambda)=\phi_i(\lambda,0)$.

In terms of the initial eigenmatrix $\Phi_0$, we can find an
explicit form of $\Phi(\mathbf{t})$ for the indefinite Toda lattice
hierarchy with $\mathbf{t}=(t_2,t_4,\ldots,t_{2n})$ (see Theorem 1
in \cite{KoYe96} and Theorem 2 in \cite{KM:96}):
\begin{Lemma}\label{phi}
The eigenvector $\phi(\lambda, \mathbf{t})=(\phi_1(\lambda,\mathbf{t}),\ldots,\phi_n(\lambda,\mathbf{t}))^T$ can be expressed as
\begin{equation} \label{tau-det-form}
\phi_k(\lambda,\mathbf{t})=\frac{c_{k}^{1/2}e^{\xi(\lambda,\mathbf{t})}}{\left[\tau^{\rm
Toda}_{k}(\mathbf{t})\,\tau^{\rm
Toda}_{k-1}(\mathbf{t})\right]^{1/2}} \left|\begin{array}{cccc}
\,m_{1,1}(\mathbf{t})  & \cdots & m_{1,k-1}(\mathbf{t})& \phi_1^0(\lambda) \\
\,m_{2,1}(\mathbf{t}) & \cdots & m_{2,k-1}(\mathbf{t}) & \phi_2^0(\lambda)\\
\vdots  &  \ddots  &   \vdots & \vdots  \\
\,m_{k,1}(\mathbf{t}) &\cdots & m_{k,k-1}(\mathbf{t}) & \phi_{k-1}^0(\lambda)
\end{array}\right|\,,
\end{equation}
where $m_{i,j}(\mathbf{t}):=\langle\phi_i^0\phi_j^0e^{2\xi(\lambda,\mathbf{t})}\rangle_{\rm Toda}$
with $\xi(\lambda,\mathbf{t})=\sum_{k=1}^n\lambda^kt_{2k}$, and
the $\tau$-functions are given by
\begin{equation}\label{indefinite-tau}
\tau^{\rm Toda}_{k} ={\rm det}\left( M^{\rm Toda}_k \right) =
\left|(m_{i,j})_{1\le i,j \le k}\right|\,.
\end{equation}
\end{Lemma}
\begin{Proof}
 With the factorization $e^{\xi(\tilde{T}_0,\mathbf{t})}=r(\mathbf{t})h(\mathbf{t})$, we have $\tilde{T}(\mathbf{t})=h(\mathbf{t})\tilde{T}_0h(\mathbf{t})^{-1}$ (see (\ref{Tequation})). Then from Lemma \ref{Torthogonal}, we obtain
 \[
 \Phi(\mathbf{t})=h(\mathbf{t})\Phi_0\,.
 \]
 Now from the factorization with $\tilde{T}\Phi=\Phi D$, we have
 \[
 \Phi(\mathbf{t})=r(\mathbf{t})^{-1}\,e^{\xi(\tilde{T}_0,\mathbf{t})}\,\Phi_0=r(\mathbf{t})^{-1}\,\Phi_0\,e^{\xi(D,\mathbf{t})}\,,
 \]
 which implies
 \begin{equation}\label{span}
 \phi_i(\lambda,\mathbf{t})\in{\rm Span}\left\{\phi_1^0(\lambda)e^{\xi(\lambda,\mathbf{t})},\ldots,\phi^0_i(\lambda)e^{\xi(\lambda,\mathbf{t})}\right\},\quad i=1,\ldots,n\,.
 \end{equation}
 Then using the Gram-Schmidt orthogonalization method with the inner product (\ref{Tproduct}), we obtain the result.
 \end{Proof}

If we set $\phi_1^0(\lambda)=1$ and consider the semi-infinite lattice with $n=\infty$,
$\phi_k(\lambda,\mathbf{t})$  can be expressed in the following elegant form
in terms of only the $\tau$-functions:
\begin{Proposition}\label{prop4.1}
The orthogonal eigenfunctions $\phi_k(\lambda)$ can be expressed in
terms of $\tau$-functions,
\[ \phi_k(\lambda, \mathbf{t}) = \frac{c_{k}^{1/2}e^{\xi(\lambda,\mathbf{t})}}{\left[\tau_{k}^{\rm{ Toda}}(\mathbf{t})\,
  \tau_{k-1}^{\rm{Toda}}(\mathbf{t})\right]^{1/2}}\,
\tau_{k-1}^{\rm{ Toda}}\left( \mathbf{t} -
   \frac{1}{2} [\lambda^{-1}]\right) \lambda^{k-1} \,,
\]
where $\xi(\lambda,\mathbf{t})=\sum_{k=1}^{\infty}\lambda^kt_{2k}$ and  $\tau_j^{\rm{Toda}}\left(\mathbf{t} - \frac{1}{2}
[\lambda^{-1}]\right) = \tau_j^{\rm{ Toda}}\left( t_2 -
\frac{1}{2\lambda}, \,\dots, \,t_{2n}-\frac{1}{2n\lambda^n},\ldots\right) $.
\end{Proposition}
\begin{Proof}
This proposition appears in \cite{adler:95} for the case $c_k = 1$.  In this more general case the proof follows from formula (\ref{tau-det-form}) with $\xi(\lambda,\mathbf{t})=\sum_{k=1}^{\infty}\lambda^kt_{2k}$: First we note that using (\ref{span}), one can replace $\phi_j^0(\lambda)$ in (\ref{tau-det-form}) by $\lambda^{j-1}$ with $\phi^0_1(\lambda)=1$, which gives $m_{i,j}(\mathbf{t})=\langle \lambda^{i+j-2}e^{2\xi(\lambda,\mathbf{t})}\rangle_{\rm Toda}$. This implies that $m_{i,j}(\mathbf{t})=\partial^{i+j-2} m_{1,1}(\mathbf{t})/\partial t_1^{i+j-2}$, hence $\tau^{\rm Toda}_k(\mathbf{t})$ is given by the Hankel determinant form,
$\tau^{\rm Toda}_k=|(h_{i+j-1})_{1\le i,j\le k}|$ with $h_{i+j-1}:=m_{i,j}$.
Since $h_j(\mathbf{t})$ is a linear combination of the exponential function $E_i(\mathbf{t})=e^{2\xi(\lambda_i,\mathbf{t})}$ and $E_i(\mathbf{t}-\frac{1}{2}[\lambda^{-1}])=(1-\frac{\lambda_i}{\lambda})E_i(\mathbf{t})$, we have
\[
h_{j}\left(\mathbf{t}-\frac{1}{2}[\lambda^{-1}]\right)=h_j(\mathbf{t})-\frac{1}{\lambda}h_{j+1}(\mathbf{t})\,.
\]
 Then it is easy to see that
\[
\tau_k^{\rm Toda}\left(\mathbf{t}-\frac{1}{2}[\lambda^{-1}]\right)\lambda^k
= \left|\begin{array}{ccccc}
h_1(\mathbf{t}) & h_2(\mathbf{t}) & \cdots & h_k(\mathbf{t}) & 1 \\
h_2(\mathbf{t}) & h_3(\mathbf{t}) & \cdots & h_{k+1}(\mathbf{t}) & \lambda \\
\vdots     &     \vdots  &  \ddots     &   \vdots    &     \vdots  \\
h_{k}(\mathbf{t}) &  h_{k+1}(\mathbf{t}) &  \cdots  &  h_{2k-1}(\mathbf{t}) & \lambda^{k-1}\\
h_{k+1}(\mathbf{t}) & h_{k+2}(\mathbf{t})& \cdots  &h_{2k}(\mathbf{t}) &\lambda^k
\end{array}\right|\,,
\]
which gives the determinant in (\ref{indefinite-tau}) in terms of $\tau_{k-1}^{\rm Toda}$.

 For a finite $n$, one can see that  $\tau^{\rm Toda}_n(\mathbf{t}-\frac{1}{2}[\lambda^{-1}])\lambda^n/\tau^{\rm Toda}_n(\mathbf{t})$
is proportional to the polynomial $\prod_{k=1}^n(\lambda-\lambda_k)$, and we have $\phi_{n+1}(\lambda)=0$, which is just the characteristic polynomial of $L$.
\end{Proof}

With Proposition \ref{prop4.1}, we will show a further example of the close
relation between the orthogonal functions in the indefinite Toda lattice
and the skew-orthogonal functions in the Pfaff lattice for the symplectic matrix of (\ref{BG.L0form}).

\subsection{The $\tau$-functions of the Pfaff and Toda lattices}

We show here that the $\tau$-functions which
generate the solutions of the Pfaff lattice equations are equivalent
to the $\tau$-functions which generate the solutions of the
indefinite Toda lattice hierarchy.
Let us recall that the $\tau$-functions of the Pfaff lattice are defined by (\ref{tau}), i.e.
\[
\tau_{2k}={\rm pf}\,(M_{2k})\,,
\]
where $M_{2k}$ is the $2k\times 2k$ upper left submatrix of $2n\times 2n$ skew-symmetric matrix $M$ given by
\[
M:={g} \,J\,{g}^T\qquad {\rm with}\qquad {g}=e^{\xi(L^2_0,\mathbf{t})}\,,
\]
where $\xi(L_0^2,\mathbf{t})=\sum_{k=1}^{n} t_{2k}L_0^{2k}$ with the initial matrix $L(0)=L_0$.
Then our goal is to show the following Theorem:
\begin{Theorem}\label{PTtau}
The $\tau$ functions of the even-flows of the Pfaff lattice hierarchy with the initial matrix $L(0)$
in the $H^{S}$-tridiagonal form of (\ref{BG.L0form}) are related to the $\tau$-functions of
the indefinite Toda lattice hierarchy by
\[ \tau_{2k}= \frac{1}{c_1\cdots c_k}\,\, \tau_k^{\rm{Toda}}\,. \]
\end{Theorem}

\begin{Proof}
We first note from (\ref{L2Tij}) that the entries of $g=e^{\xi(L_0^2,\mathbf{t})}$ are expressed by those of
$\tilde{g}=e^{\xi(\tilde{T}_0,\mathbf{t})}$,
\[
g=\tilde{g}\otimes\left(\begin{array}{cc}1&0\\0&0\end{array}\right)\,+\,
\tilde{g}^T\otimes \left(\begin{array}{cc}0&0\\0&1\end{array}\right)\,.
\]
Since $\tilde{T}_0=CT_0$ with a symmetric matrix $T_0$,
$C^{-1}\tilde{g}$ is symmetric, i.e. $C^{-1}\tilde{g}=\tilde{g}^TC^{-1}$.
Then for $H = C^{-1}M^{\rm{Toda}}=C^{-1} \tilde{g} C \tilde{g}^T=\tilde{g}^T\tilde{g}^T$ and  $M = {g} J {g}^T$,  we have
\begin{equation}\label{mij}
M=H\otimes \left(\begin{array}{cc}0&1\\0&0\end{array}\right)\,+\,
H^T\otimes \left(\begin{array}{cc}0&0\\-1&0\end{array}\right)\,.
\end{equation}

Let $H_k$ be the $k\times k$ upper left submatrix of $H$, and $M_{2k}$ be the submatrix defined above.
 Then with the permutation matrix
\[ P = \left[ e_1, e_3, \dots, e_{2k-1}, e_{2}, e_4, \dots, e_{2k}
\right]\,,\]
where $e_k$'s are the standard basis (column) vectors on $\mathbb{R}^{2k}$,
one can transform $M_{2k}$ to
\begin{equation}\label{factor}
 P^{T} M_{2k} P = \left(\begin{array}{cc} 0_k &  H_k \\ -H_k^T &
0_k \end{array}\right) =
\left(\begin{array}{cc} H_k & 0_k \\ 0_k & I_k \end{array}\right) P^T J P
\left(\begin{array}{cc} H_k^T & 0_k \\ 0_k & I_k \end{array}\right)
\,. \end{equation}
We recall the Pfaffian identity for a skew-symmetric matrix $B$,
\[ \mbox{pf}( A^T B A ) = \mbox{det}(A) \mbox{pf}(B) \,, \]
which with (\ref{factor}) implies that
\[ \mbox{det}(P) \mbox{pf}(M_{2k}) = \mbox{det}(P)
\mbox{det}\left(\begin{array}{cc} H_k & 0_k \\ 0_k & I_k \end{array}\right)
\mbox{pf}(J)\,. \]
We finish the proof by noting that $P$ is invertible and $\mbox{pf}(J)
= 1$ giving
\[ \mbox{pf}(M_{2k}) = \mbox{det}\left(\begin{array}{cc} H_k & 0_k \\ 0_k &
  I_k \end{array}\right) = \mbox{det}( H_k)=\frac{1}{c_1\cdots c_k}\,{\rm det}(M_k^{\rm Toda})\,. \]
\end{Proof}

\subsection{Skew-orthogonal and orthogonal functions}
We show here that the eigenvectors of the matrix $L$ in the form
(\ref{BG.L0form}) define a family of skew-orthogonal functions,
and then discuss an explicit relation between those skew-orthogonal
functions and the orthogonal functions appearing in the
indefinite Toda lattice, i.e. those in Proposition \ref{prop4.1}.

As in the case of the Toda lattice, the Pfaff lattice hierarchy (\ref{pfaffhierarchy})
is given by the compatibility of the linear equations,
\begin{equation}\label{LBequations}
L\Psi=\Psi\Lambda \qquad{\rm and}\qquad \frac{\partial \Psi}{\partial t_{2j}} =\pi_{\mathfrak{sp}}(L^{2j})\,\Psi\,,
\end{equation}
where $L$ is symplectic and in the form of (\ref{BG.L0form}), and
$\Lambda$ is the eigenvalue matrix,
\[
\Lambda:=D\otimes \left(\begin{array}{cc}-1&0\\0&1\end{array}\right)\,,
\]
with $D={\rm diag}(z_1,\ldots,z_n)$. (Recall that the
eigenvalues of $L$ consist of the pairs $(z_k,-z_k)$ for
$k=1,\ldots,n$, and we assume here that they are all distinct.) Then
we have:
\begin{Lemma}\label{SOrelation}
The eigenmatrix $\Psi$ satisfies the skew-orthogonal relation,
\[
\Psi^T\,J\,\Psi=\mathcal{K}\,J\,,
\]
with a diagonal matrix with nonzero constants $\kappa_k$'s in the form,
\[
\mathcal{K}={\rm diag}_2\left(\kappa_1I_2,\kappa_2I_2,\ldots,\kappa_nI_2\right)\,.
\]
\end{Lemma}
\begin{Proof}
  From (\ref{LBequations}) with $B:=\pi_{\mathfrak{sp}}(L^{2j})$ and $t=t_{2j}$, we have
\begin{eqnarray*}
\frac{\partial}{\partial t}(\Psi^TJ\Psi ) &= \Psi^TB^TJ\Psi+\Psi^TJB\Psi \\
    &= \Psi^T(B^TJ+JB)\Psi=0\,,
\end{eqnarray*}
where we have used $B\in \mathfrak{sp}(n)$.

Now we note the following equation for the eigenvectors $\psi(z_k)$,
\begin{eqnarray*}
z_j\psi(z_k)^TJ\psi(z_j) &= \psi(z_k)^TJL\psi(z_j) \\
        &=- \psi(z_k)^TL^TJ\psi(z_j)\,  \qquad (\because~ JL+L^TJ=0) \\
        &= -(L\psi(z_k))^TJ\psi(z_j) \\
        &=-z_k\psi(z_k)^TJ\psi(z_j)\,.
\end{eqnarray*}
This implies
\[
(z_k+z_j)\psi(z_k)^TJ\psi(z_j)=0\,.
\]
Noting the order of the eigenvalues in $\Lambda={\rm
diag}(-z_1,z_1,\ldots, -z_n,z_n)$, so that $\Psi^T J \Psi$, which is
skew-symmetric, is a multiple of  $J$ of the form $\mathcal{K}$ with nonzero
constant $\kappa_k$'s.
\end{Proof}
 From Lemma \ref{SOrelation}, we have
 \begin{equation}\label{SOIrelation}
 \Psi\,\mathcal{K}^{-1}\,J\,\Psi^T=J\,,
 \end{equation}
 which implies the skew-orthogonality relation for the functions $\psi_k(z)$ in the eigenvector
 $\psi(z)=(\psi_0(z),\psi_1(z),\ldots,\psi_{2n-1}(z))$ of $L$,
\begin{equation}\label{SO}
\left\{\begin{array}{lllll}
\langle \psi_{2j}, \psi_{2k} \rangle=\langle \psi_{2j+1}, \psi_{2k+1} \rangle=0,\\[0.4ex]
  \langle \psi_{2j},\psi_{2k+1} \rangle=\langle\psi_{2j-1},\psi_{2k}\rangle =\delta_{jk},  \end{array}\right.
   \quad
  {\rm for}\quad  j\le k\,.
\end{equation}
Here the inner product $\langle f,g\rangle$ is defined by
\begin{equation}\label{skew-inner}
\langle f, g\rangle:=\sum_{k=1}^n\left[f(-z_k)g(z_k)-f(z_k)g(-z_k)\right]\kappa_k^{-1}\,.
\end{equation}

Now we claim:
\begin{Theorem}\label{SO-O relation}
The skew-orthogonal eigenfunctions
 $\psi_k(z,\mathbf{t})$ satisfying (\ref{SO}) can be given by
\[\left\{\begin{array}{lll}
\psi_{2k}(z,\mathbf{t})=\phi_{k+1}(\lambda,\mathbf{t})\\
{}\\
\psi_{2k+1}(z,\mathbf{t})= c_{k+1}^{-1}
z\,\phi_{k+1}(\lambda,\mathbf{t})
\end{array}\right. \qquad {\rm with}\quad \lambda=z^2\,,
\]
where $\phi_k(\lambda,\mathbf{t})$'s are the orthogonal eigenfunctions appearing in the
indefinite Toda lattice (see Proposition \ref{prop4.1}), and the measure of the inner
product $\langle\cdot,\cdot\rangle$ is given by $\kappa_k=2z_kc_k^{-1}$.
\end{Theorem}

Before proving the Theorem, we recall the following Proposition which gives
the skew-orthogonal  functions
appearing in the semi-infinite Pfaff lattice:
\begin{Proposition}\label{prop4.2}
The eigenvector $\psi(z,\mathbf{t})=(\psi_0(z,\mathbf{t}),\,\psi_1(z,\mathbf{t}),\,\ldots)^T$ for the matrix $L$ of (\ref{L}) with $n=\infty$ can be expressed in
terms of $\tau$-functions,
\[
\left\{ \begin{array} {llll}
\psi_{2k}(z, \mathbf{t}) =\displaystyle{
\frac{e^{\xi(z,\mathbf{t})}}{\left[\tau_{2k}(\mathbf{t})
\tau_{2k+2}(\mathbf{t})\right]^{1/2}}
\tau_{2k}(\mathbf{t}- [z^{-1}] ) z^{2k}} \\
{}\\
\psi_{2k+1}(z, \mathbf{t}) = \displaystyle{ \frac{c_{k+1}^{-1}e^{\xi(z,\mathbf{t})}
}{\left[\tau_{2k}(\mathbf{t})
  \tau_{2k+2}(\mathbf{t})\right]^{1/2}} \left( z +
\frac{\partial}{\partial t_1} \right) \tau_{2k}( \mathbf{t} -
     [z^{-1}]) z^{2k}\,,}
\end{array} \right.
\]
where $\xi(z,\mathbf{t})=\sum_{k=1}^{\infty}z^kt_k$ and $\tau_{2k}(\mathbf{t} - [z^{-1}]) =
\tau_{2k}(t_1 - \frac{1}{z},\, t_2 - \frac{1}{2z}, \,\dots )$.
\end{Proposition}
This Proposition appears as Theorem 3.2 in \cite{adler:02} for the case $c_k
=1$.  In this more general case the additional factor of a $c_{k+1}^{-1}$ is
present to ensure that the recursion relation $L$ has $L_{2k+1, 2k} =
c_{k+1}$.  This entry of the Pfaff variable is a Casimir of the Pfaff lattice
equations and so is fixed.  In the definition of the skew-orthogonal
functions
 it corresponds to a choice of the ratio between leading coefficients
 of the polynomial parts of $\psi_{2k}$ and $\psi_{2k+1}$.

Now we prove Theorem \ref{SO-O relation}:
\begin{Proof}
 First we show that $\psi_k(z,\mathbf{t})$'s in the Theorem are the same as those in Proposition
 \ref{prop4.2}, when the matrix $L$ is in an $H^S$-tridiagonal form (\ref{BG.L0form}). To show this, we recall
 Theorem \ref{PTtau}, i.e.
 \[
 \tau_{2k}(\mathbf{t})=\frac{1}{c_1\cdots c_k}\,\tau^{\rm Toda}_k(\mathbf{t})\,.
 \]
Since the odd flows are trivial, we have $\mathbf{t}=(t_2,t_4,\ldots, t_{2n})$. Then substituting
this relation into $\psi_k(z,\mathbf{t})$ in Proposition \ref{prop4.2}, it is straightforward to show the equations
in the Theorem.

Now we show the skew-orthogonality relation (\ref{SOrelation}): We only check the case
$\langle \psi_{2j}, \psi_{2k+1}\rangle=\delta_{jk}$, and the others are trivial with the form of
the inner product.
\begin{eqnarray*}
\langle\psi_{2j},\psi_{2k+1}\rangle &= \sum_{i=1}^{n}\left[\psi_{2j}(-z_i)\psi_{2k+1}(z_i)
-\psi_{2j}(z_i)\psi_{2k+1}(-z_i)\right]\frac{c_i}{2z_i} \\
    & = c_{k+1}^{-1} \sum_{i=1}^n \phi_{j+1}(\lambda_i)\phi_{k+1}(\lambda_i) c_i = \delta_{jk}\,,
\end{eqnarray*}
where we have used the orthogonal relation $\Phi C\Phi^T=C$ in Lemma \ref{Torthogonal}.
\end{Proof}

We also note that the moment matrix $M(\mathbf{t})$ for the Pfaff lattice has the similar form
as the $M^{\rm Toda}(\mathbf{t})$ for the indefinite Toda lattice given in (\ref{Mtoda}),
\begin{eqnarray*}
M(\mathbf{t}) &=e^{\xi(L^2_0,\mathbf{t})}\,J \, e^{\xi(L^2_0,\mathbf{t})^T}
      = \Psi_0\,e^{\xi(\Lambda^2,\mathbf{t})}\,\Psi_0^{-1}\,J\,\Psi_0^{-T}\,e^{\xi(\Lambda^2,\mathbf{t})}\,\Psi_0^T\\
      &= \Psi_0\,e^{2\xi(\Lambda^2,\mathbf{t})}\,\mathcal{K}^{-1}\,J\,\Psi_0^T
      = \left(\langle \psi_i^0,\psi_j^0\,e^{2\xi(z^2,\mathbf{t})}\rangle\right)_{0\le i,j\le 2n-1}\,.
\end{eqnarray*}
Comparing with (\ref{mij}), the entries $m_{ij}:=\langle\psi^0_i,\psi^0_je^{2\xi(z^2,\mathbf{t})}\rangle$ are given by $m_{ij}=-m_{ji}$ with
\[
m_{2j-1,2k}=-m_{2k,2j-1}\qquad {\rm for}\quad j\le k\,,
\]
and zero for all other cases, i.e. $m_{2j-1,2k-1}=m_{2j,2k}=0$. With this structure of the moment matrix, one can
give a direct proof of Theorem \ref{SO-O relation} without Proposition \ref{prop4.2}, however
this might be a less elegant approach.

\begin{Remark}
The skew-inner product of (\ref{skew-inner}) is closely related to the
skew-inner product for GSE-Pfaff lattice, which is defined as
\[
\langle f,g\rangle_{\rm GSE}:=\sum_{k=1}^n\{f,g\}(z_{2k})c_k\,,
\]
where $\{f,g\}(z)=f'(z)g(z)-g'(z)f(z)$ with $f'(z)=df(z)/dz$.
This inner product may be obtained by the following
inner product in the limit $z_{2k-1}\to z_{2k}$ (see \cite{kodama:07}), i.e.
\[
\langle f,g\rangle_{\rm GSE}=\lim \sum_{k=1}^n
\frac{f(z_{2k-1})g(z_{2k})-f(z_{2k})g(z_{2k-1})}{z_{2k}-z_{2k-1}} c_k\,,
\]
In the case of (\ref{skew-inner}), we instead take the limit $
z_{2k-1}\to -z_{2k}$.

Also in the case of the continuous measure, the inner product
(\ref{skew-inner}) can be written by
\[
\langle f, g \rangle= \int_{\Sigma} \left[ f(-z)
g(z) - f(z) g(-z) \right]  \,c(z^2)\,dz\,,
\]
where  $\Sigma = \mathbb{R}_+ \cup i \mathbb{R}_+$ oriented
from $i\infty$ to $0$ and then to $\infty$.
To reduce this expression back to the discrete case we take
\begin{eqnarray*}
c(z^2)\,  dz \big|_{\Sigma} &= \sum_{k=1}^{n} \delta(z^2 - z_k^2) \, dz
\Big|_{\Sigma}  = \sum_{k=1}^n \delta( (z-z_k) (z+z_k) ) dz \Big|_{\Sigma}\\
&= \sum_{k=1}^n \left[ \frac{1}{2 z_k} \delta( z- z_k) + \frac{1}{2
z_k} \delta(z+z_k) \right] dz \Big|_{\Sigma} = \sum_{k=1}^n
\frac{1}{2 z_k} \delta(z-z_k)\, dz \Big|_{\Sigma}\,, \end{eqnarray*}
with $z_k \in \Sigma$, and where we used the property $\delta( a z)
= \frac{1}{|a|} \delta(z)$ for $a\in\mathbb{R}$ and $\delta( a z) =
\frac{1}{i |a|} \delta(z) $ for $a\in i \mathbb{R}$. The last
equality is because we have restricted to the positive and positive
imaginary square roots with $z_k \in \Sigma$. The result is that
this discrete inner product will agree with (\ref{skew-inner}).

\end{Remark}

\subsection{Asymptotics of the even Pfaff flows}

In this final section, we mention the asymptotic behavior of the even Pfaff lattice
 based on the results of the indefinite Toda lattice discussed in \cite{KoYe96, KoYe98}:
Theorem \ref{TP} implies that the Pfaff flow on $H^S$-tridiagonal matrix
 $L_0$ with $c_k = \pm 1$ is equivalent to the indefinite
Toda flow on $\tilde{T} = C T$ given by (\ref{T}).  This system was
studied in detail in \cite{KoYe96} and \cite{KoYe98}.  It is a
version of the Toda lattice which uses HR-factorization in place
of QR-factorization, see \cite{Fa93}.  The goal of HR-factorization is to
 write an element $g \in SL(n, \mathbb{R})$ as $g
= r h $ where $r$ is lower triangular and $h$ satisfies
\[ h C h^T = C\,. \]
For $g(t_{2j}) := \exp({ t_{2j} \tilde{T}(0)^j })$, the HR-factorization $g(t_{2j}) = r(t_{2j}) h(t_{2j}) $
gives the solution of the $t_{2j}$-flow of the Pfaff lattice,
 \[\tilde{T}(t_{2j}) = r^{-1}(t_{2j}) \tilde{T}(0)
r(t_{2j}) = h(t_{2j}) \tilde{T}(0) h^{-1}(t_{2j})\,. \]
Thus the indefinite Toda lattice is a continuous version of the
HR-algorithm \cite{Fa93}.

The indefinite Toda flow may experience a blow up, where some of the
entries reach infinity in finite time \cite{KoYe98}.  A blow up occurs when one of
the $\tau$ functions becomes $0$.  In our case this is precisely
when one of the $\tau_{2k} = 0$. The initial conditions for a blow-up are characterized by
\begin{Theorem}[Theorem 3 in \cite{KoYe96}]
If $\tilde{T}(0)$ possess non-real eigenvalues or non-real eigenvectors while $\tau_{2k}(0)
\neq 0$, then $\tilde{T}(t)$ blows up to infinity in finite time.
\end{Theorem}
The second possibility occurs if the $c_k$ do not all have the
same sign. The case when the eigenvalues of $\tilde{T}$ are real and
the $c_k$ do not all have the same sign is studied in Section 4 of
\cite{KoYe98}. It is shown that every flow (considered in both
directions $t_{2j}\to \pm \infty$) contains a blow up.  We may then
compactify the flows by adding infinite points representing the blow
ups.  It is then shown that the fixed points of the compactified
flows are diagonal matrices.   In other words $a_k \to 0$ and $ d_k \to
c_k z^2_{\sigma(k)}$ as $t_{2j}\to\infty$ for some permutation $\sigma\in\mathcal{S}_n$,
the symmetric group of order $n$.
 When the eigenvalues of $\tilde{T}$ are real we may count the
total number of blow ups on the flow.  If the number of changes of
sign in the sequence $c_k$ is $m$ then the number of blow ups is $m
(n-m)$ \cite{KoYe98}.

There is a condition under which the indefinite Toda flow, with $c_k
= - 1$ for some $k$,  does not contain a blow up in the positive $t$
direction.  A lower triangular matrix is called lower triangular
totally positive if all its non-trivial minors are positive.
 It was shown in \cite{gekhtman:97} that if the lower component of
the LU-factorization of the eigenvector matrix of $\tilde{T}_0$ is
lower triangular totally positive then there are no blow ups on the
indefinite Toda flow for $t > 0$.

If $L_0$ has only real and imaginary eigenvalues, then $\tilde{T}$ has
real eigenvalues.  By Theorem
\ref{TP} we find that the flow may be continued through a blow up, if
it occurs, and that this compactified flow converges:
\[ L(t_{2j}) \to \mbox{diag}_2\left( \left(\begin{array}{cc} 0 & c_1
  \\ c_1 z_{\sigma(1)}^2 & 0 \end{array}\right), \dots,
\left(\begin{array}{cc} 0 & c_n \\ c_n z_{\sigma(n)}^2 & 0 \end{array}\right)
\right) \,,
\]
where the eigenvalues of $L_0$ are $\{\pm z_m:m=1,\ldots,n\}$.

If the $c_k=1$ for all $k$, then $\tilde{T}$ is symmetric and
no blow up is possible.  We will restrict to this case from now on.
We appeal to the abundant literature on the Toda lattice for the
results needed \cite{deift:83, KoYe96, KoYe98}. In
particular, under mild assumptions, as $t_{2j} \to \pm\infty$,
$T(t_{2j})$ given by (\ref{Ttoda}) converges to a diagonal matrix.
As a result we see that $a_k(t_{2j}) \to 0$.

In terms of the Pfaff lattice this is showing that  $L(t_2)$
converges to a $2\times 2$ block diagonal matrix as
$t_2\to \infty$ where the
$j$-th diagonal block has the form
\[ \left(\begin{array}{cc} 0 & 1 \\ z_j^2 & 0 \end{array}\right) \,. \]
The blocks will be sorted by the size of $z^2$ so that the blocks
with pairs of imaginary eigenvalues will appear in the lower right
corner while those with pairs of real eigenvalues will appear in the
upper left corner.
For the $t_4$-flow, the eigenvalues will again be such a pair but will
now be sorted by the size of $z^4$, mixing the blocks with real
and imaginary pairs.

While the end result is that the Pfaff lattice restricted to symplectic lower Hessenberg matrices is equivalent in every sense to the indefinite Toda lattice flows, this is not a result which is immediately obvious from the Pfaff lattice itself.  One sees that even in the case when the $c_k=1$ for all $k$, the long time dynamics of the Pfaff lattice flows of symplectic lower Hessenberg matrices is significantly different from those of non symplectic lower Hessenberg matrices; for example the the symplectic matrices are in fact fixed points of the flow.  This difference alone merited a separate work on these cases and produced some interesting and new connections to generalizations of the Toda lattice.

\appendix
\addcontentsline{toc}{section}{Appendix: SR-algorithm}
\section*{Appendix: SR-algorithm}
\setcounter{section}{1}

The Toda lattice may be viewed as a continuous version of the QR-algorithm for diagonalizing a symmetric tridiagonal matrix.  In the same way the Pfaff lattice is a continuous version of the SR-algorithm on an  $H^S$-tridiagonal matrix.  Here we collect some pertinent details about the SR-algorithm.  These results have been covered from a variety of points of view, a partial list of references is \cite{benner:97, bunse:86B, Fa00, rutishauser:58, symes:82, watkins:84, watkins:88, watkins:90}.

Here we consider the SR-algorithm for a symplectic matrix which is
symplectically similar to a lower Hessenberg matrix. Recall that any matrix can be
reduced to a similar Hessenberg matrix by Householder's method, a fact
 used in giving an efficient version of the
QR-algorithm. However, note that the Householder's method does not, in
general, give a symplectic
conjugation.

We will need Theorem 3.4 from \cite{bunse:86B}:
\begin{Theorem}\label{BG.hessenbergL0}
Let $h\in Sp(n,\mathbb{R})$ and $s_1$ be a row vector in $\mathbb{R}^{2n}$.  Then there exists a symplectic transformation $S$ such that
$S h S^{-1}$  is a lower Hessenberg symplectic matrix
iff $K(h, s_1)$ has an SR-factorization $K(h, s_1) = RS$
 where $R$ is a lower triangular matrix and
$s_1$ is the first row of $S\in Sp(n,\mathbb{R})$.  In addition, if this factorization
exists, then
\[ L_0 := S h S^{-1} = R^{-1} C_h R \]
is in the lower Hessenberg form,
\begin{equation}
L_0=
\left(\begin{array}{ccccccc}
\begin{array}{cc} 0 & c_1 \\ d_1 & 0 \end{array} &\vline&
\begin{array}{cc} 0  & 0 \\ a_1 & 0 \end{array} &\vline& \cdots &\vline& 0_2
\\\hline\bigxstrut
\begin{array}{cc} 0 & 0 \\ a_1 & 0 \end{array} &\vline&
\begin{array}{cc} 0 & c_2 \\ d_2 & 0 \end{array} &\vline&\cdots &\vline& 0_2
\\\hline\bigxstrut
\vdots &\vline& \vdots &\vline& \ddots &\vline& \vdots
\\\hline\bigxstrut
0_2 & \vline& 0_2 &\vline& \cdots &\vline &\begin{array}{cc} 0 & c_n \\ d_n & 0\end{array}
\end{array}\right) \,,
\end{equation}
that is, $L_0$ has a $2\times 2$ block tridiagonal form
$ L_0 = (l_{i,j})_{1\le i,j\le n} $ with
 $2\times 2$ block matrices $l_{i,j}$ having $l_{i,j}=0_2$ for $|i-j|>1$ and
\[
l_{k,k}=\left(\begin{array}{cc} 0&c_k \\d_k & 0\end{array}\right),\quad
l_{k,k+1}=l_{k+1,k}=\left(\begin{array}{cc} 0&0\\a_k&0\end{array}\right)\,,
\]
where $c_k = \pm 1$.
\end{Theorem}
A dense set of symplectic matrices $h$ may be placed in this form.
We call a matrix in the form (\ref{BG.L0form}) an
``$H^S$-tridiagonal'' matrix, which plays the similar role as a tridiagonal matrix
in the case of the symmetric matrices.

Furthermore, at the expense of excluding a large set of symplectic
matrices, we may refine Theorem \ref{BG.hessenbergL0}:
\begin{Theorem}\label{hessenbergL0}
Let $h\in Sp(n,\mathbb{R})$, and $s_1$ be a row vector in $\mathbb{R}^{2n}$. Then
there exists a symplectic transformation $S$ such that $S h
S^{-1}$ is a lower Hessenberg matrix iff $K(h, s_1)$ has an
SR-factorization $K(h, s_1)=RS$ with $R\in G_{\mathfrak{k}}$, and
 with $s_1$ equal to the first row of $S$.
In addition, if this factorization exists, then
\[ L_0 := S h S^{-1} = R^{-1} C_h R \]
is in an $H^S$-tridiagonal form  (\ref{BG.L0form}) with $c_k =1$ for
all $k$.
\end{Theorem}

We now show that if $h$ has quadruples of complex eigenvalues
then it will not be similar to a matrix in the $H^S$-tridiagonal form of
(\ref{BG.L0form}) with $c_k = 1$ for all $k$:
\begin{Proposition}
Let $L$ be a $2n\times 2n$ matrix in the form of (\ref{BG.L0form})
with $c_k =1$ for all $k$.  The
eigenvalues of $L$ come in real or imaginary pairs without multiplicity.
\end{Proposition}
\begin{Proof}
One checks that for $L$ in this form, we have
\[ L^2 = T \otimes I_2 \]
where $T$ is a symmetric tridiagonal $n\times n$ matrix given by
\begin{equation}\label{Tsymm}
T=\left(\begin{array}{ccccc}
d_1 & a_1 & 0 & \cdots & 0 \\
a_1 & d_2 & a_2 &\cdots & 0 \\
\vdots&\vdots&\ddots &\ddots& \vdots\\
0    &   0  &  \cdots &  d_{n-1} & a_{n-1} \\
0   &    0  &  \cdots  &  a_{n-1} & d_n
\end{array}\right)
\end{equation}
  From this one computes that
\[ F(\lambda) = \mbox{det}\left( L^2 - \lambda I_{2n} \right) = \left[ \mbox{det}\left( T -
\lambda I_n \right) \right]^2 \,.\]
As $T$ is symmetric, $F(\lambda)$ only has $n$ real roots each with
multiplicity $2$, therefore the
eigenvalues of $L$ are square roots of real numbers, and so are only
real or imaginary.
\end{Proof}
There is an algorithm for carrying out the transformations of Theorems \ref{BG.hessenbergL0} and \ref{hessenbergL0}, see \cite{Fa00}.


This procedure is  the analogous operation to the QR-algorithm step of
transforming a symmetric matrix to a tridiagonal form  by the
Householder's method. There is a large time savings in carrying out
this transformation first as an SR-factorization of a matrix in the
form of  $L_0$ only requires $\mathcal{O}(n)$ operations.

The SR-algorithm is defined as an iteration with initial matrix $L_0$ a
symplectic matrix
and recursion given by factoring $L_{k-1} = R_k S_k $ using the
SR-factorization then taking $L_{k} := S_kR_k=S_k L_{k-1} S_k^{-1} =
R_k^{-1} L_{k-1} R_k $.  As each step is a similarity transform of the
previous step by both a matrix in the symplectic group and a matrix in
${G}_\mathfrak{k}$, $L_k$ is still a lower Hessenberg matrix
which is also symplectic, and therefore is still in the $H^S$-tridiagonal
 form of $L_0$
above.

If the algorithm is successful the
 $L_k$ approaches a family of block diagonal matrices with blocks
of the form:
\begin{enumerate}

\item[(a)] $2\times 2$ blocks containing two real eigenvalues $(z, -z)$,

\item[(b)] $2\times 2$ blocks containing two imaginary eigenvalues,
  $(z, -z)$,

\item[(c)]  $4 \times 4$ blocks containing a quadruple of complex
eigenvalues $(z, \bar{z}, -z, -\bar{z})$.

\end{enumerate}
In addition the blocks are sorted by the size of $\ln|z|$.
If there are complex eigenvalues the sequence does not converge to a
fixed matrix, rather it approaches the sorted diagonal shape.
If the $c_j=1$ for all $j$, then $L_k$ converges to a block
diagonal matrix with just $2\times 2$ blocks.  The rate of
convergence is $\mathcal{O}(k^3)$ for a dense set of $L_0$.  In
practice
 one runs the algorithm until $\sup_j | a_j(k)|$  at the $k$th step is less than some fixed
$\epsilon$ tolerance. The algorithm also works on the initial matrix
$h$ (i.e. without the change to a lower Hessenberg $L_0$).

There is a substantial literature on improvements to this basic
algorithm using implicit SR-factorization steps on certain matrix functions of
$L_k$ rather than just $L_k$ (see e.g. \cite{Fa00} and references therein).

We now show that the SR-algorithm is directly equivalent to the Pfaff lattice flows with a non-traditional Hamiltonian.  In light of the connection to the indefinite Toda lattice and HR-algorithms this is not a surprising fact and the proof is appropriately close to that for the analogous fact in the Toda cases.
\begin{Proposition}\label{pfaffalgorithm}
Let $L_0\in \mathfrak{sp}(n)$ in the form of (\ref{L}). Then
the SR-algorithm is equal to the integer evaluations of
the Pfaff lattice flow with respect to Hamiltonian $H(L) =
{\rm{tr}}\left( L \ln(L) - L \right)$ with $L(0)=L_0$.
\end{Proposition}

\begin{Proof} Recall that the SR-algorithm for the initial matrix $L_0$
is given by
\[ L_{k-1}=R_kS_k \qquad {\rm and} \qquad L_k:=S_kR_k\,.\]
One can see that
\[ L_k = S_k S_{k-1} \cdots S_2 S_1 L_0 S_1^{-1} S_2^{-1} \cdots
S_{k-1}^{-1} S_k^{-1} \,.
\]
While the Pfaff flows arise from (\ref{LRS}),
\[ L(t) = S(t) L_0 S(t)^{-1} \]
where
\[ \exp\left( t \ln(L_0)\right) = L_0^t = R(t) S(t)\,. \]
We want to show that $S(k) = S_k S_{k-1} \cdots S_2 S_1 $.  We
prove this by induction:   First we check that $S(1) = S_1$, from
which $L_1 = L(1)$.  Next we make the inductive hypothesis that
$S(k-1) = S_{k-1} S_{k-2} \cdots S_2 S_1 $ and then consider
\begin{eqnarray*}
R(k) S(k) &= \exp\left( k \ln(L_0)\right) = L_0^k = L_0^{k-1} L_0 \\
&= R(k-1) S(k-1) L_0 \\
&= R(k-1) L_{k-1} S(k-1) \\
&= R(k-1) R_k S_k S(k-1) \,.
\end{eqnarray*}
By uniqueness of SR-factorizations for $R$ in the identity component
of $G_{\mathfrak{k}}$,  we see that $S(k) = S_k S(k-1) =
S_k S_{k-1} \cdots S_2 S_1 $.
\end{Proof}

\addcontentsline{toc}{section}{Acknowledgments}
\ack

This work was supported by NSF grant DMS0806219, in addition V.P. acknowledges the partial support of NSF-VIGRE grant DMS-0135308.


\end{document}